\documentclass[11pt]{article}
\usepackage{mathrsfs}
\usepackage{CJK}
\usepackage{bbm}
\usepackage{graphicx}
\usepackage{amsmath,amssymb}
\usepackage{amsfonts}
\usepackage{amsthm}
\usepackage{verbatim}
\usepackage{color}
\usepackage{graphicx}
\usepackage{amssymb}
\usepackage{latexsym}
\usepackage{bm}
\large\normalsize

\newcommand{\be}{\begin{equation}}
\newcommand{\ee}{\end{equation}}
\newcommand{\bes}{\begin{equation}\begin{aligned}}
\newcommand{\ees}{\end{aligned}\end{equation}}
\newcommand{\ben}{\begin{equation}\nonumber\begin{aligned}}

 \newcommand{\R}{\mathbb{R}}

\renewcommand{\leq}{\leqslant}
\renewcommand{\geq}{\geqslant}

\newtheorem{thm}{Theorem}[section]
\newtheorem{lem}[thm]{Lemma}

\newtheorem{defi}[thm]{Definition}
\newtheorem*{main thm}{Main Theorem}

\numberwithin{equation}{section}
\begin{document}

\title{A sharp recovery condition for sparse signals with partial support information via orthogonal matching pursuit}

\author{ Huanmin~Ge  and Wengu~Chen
\thanks{H. Ge is with Graduate School, China Academy of Engineering Physics,
Beijing, 100088, China, e-mail:gehuanmin@163.com.}
\thanks{W. Chen is with Institute of Applied Physics and Computational Mathematics,
Beijing, 100088, China, e-mail: chenwg@iapcm.ac.cn.}
\thanks{This work was supported by the NSF of China (Nos.11271050, 11371183)
.} }

\maketitle

\begin{abstract}
This paper considers the exact recovery of  $k$-sparse signals in
the noiseless setting and support recovery  in the noisy case when
some prior information on the support of the signals is available.
This prior support consists of two parts. One part is a subset of
the true support and another part is outside of the true support.
For $k$-sparse signals $\bm{x}$ with the prior support which is
composed of $g$ true indices and $b$ wrong indices, we show that if
the restricted isometry constant (RIC) $\delta_{k+b+1}$ of the
sensing matrix $\bm{A}$ satisfies
\begin{eqnarray*}
\delta_{k+b+1}<\frac{1}{\sqrt{k-g+1}},
\end{eqnarray*}
then orthogonal matching pursuit (OMP)  algorithm can perfectly
recover the signals $\bm{x}$ from $\bm{y}=\bm{Ax}$ in $k-g$
iterations. Moreover, we show the above sufficient condition on the
RIC is sharp. In the noisy case, we achieve the exact recovery of
the remainder support (the part of the true support outside of the
prior support) for the $k$-sparse signals $\bm{x}$ from
$\bm{y}=\bm{Ax}+\bm{v}$ under appropriate conditions. For the
remainder support recovery, we also obtain a necessary condition
based on the minimum magnitude of partial nonzero elements of the
signals $\bm{x}$.
\end{abstract}


\section{Introduction}
Compressive sensing has been a very active area of recent research
in signal processing, applied mathematics and statistics
\cite{BS,DLTB,LDSP}, \cite{HS} and \cite{TLDRB}.
 A central  aim of compressive sensing is  to reconstruct sparse signals from inaccurate and incomplete measurements. In compressive sensing,
one considers the following  model:
\begin{eqnarray}\label{model1}
\bm{y}=\bm{Ax}+\bm{v},
\end{eqnarray}
where $\bm{y}\in \R^m$ is a measurement vector, the matrix $\bm{A}\in\R^{m\times n} \ (m\ll n)$ is a known sensing matrix,
the vector $\bm{x}\in\R^n$ is an  $k$-sparse signal  and $\bm{v}\in\R^m$ is a vector of measurement errors.
In particular, $\bm{v}=\bm{0}$ in the  noiseless setting.
Denote  the support of the vector $\bm{x}$ by $T=\mathrm{supp}(\bm{x})=\{i|\bm{x}_i\neq 0\}$ and the
size of its support  with $|T|=|\mathrm{supp}(\bm{x})|$.  If
$|\mathrm{supp}(\bm{x})|\leq k$, $\bm{x}$ is called $k$-sparse.
The goal is to recover the  unknown $k$-sparse signal $\bm{x}$ from  $\bm{y}$ and $\bm{A}$ in the model \eqref{model1} using fast and  efficient algorithms.

In order to analyze  the $\ell_1$-minimization, Cand\`{e}s and Tao
\cite{CT} introduced  a commonly used framework: the restricted
isometry property (RIP).
\begin{defi}\label{definition1}
A matrix $\bm{A}$ satisfies the RIP of order $k$ if there exists a constant $\delta_k\in[0,1)$ such that
\begin{eqnarray}\label{p2}
(1-\delta_k)\|\bm{x}\|_2^2\leq\|\bm{Ax}\|_2^2\leq(1+\delta_k)\|\bm{x}\|_2^2
\end{eqnarray}
holds for all $k$-sparse signals $\bm{x}$. And the smallest constant $\delta_k$ is called the restricted isometry constant (RIC).
\end{defi}
In this paper, we focus on a kind of  sparse signals which have some
prior support information (possibly erroneous). The recovery of such
sparse signals with a strong dependence on their prior supports has
been introduced in several contributions and possess practical and
analytical interests in many setups \cite{BMP,BSV, FMSY, J, KXAH,
NSW} and \cite{VL}. For example, this type of signals occurs in
video compression or dynamic magnetic  resonance imaging where the
supports of the sought vectors commonly evolve slowly with time.

Compressed sensing  has previously been studied under different conditions
for recovering sparse signal in the presence of prior support information.
To make good use of  prior support information of the signals, the following weighted $\ell_1$ minimization has  been introduced
\begin{eqnarray}\label{pro2}
\min_{\bm{x}\in \mathbb{R}^{n}}\|\bm{x}\|_{1,\bm{\mathrm{w}}} \ \ \
 {\rm subject\quad to}\ \ \  \|\bm{y}-\bm{Ax}\|_2\leq\epsilon,
\end{eqnarray}
where $\mathrm{\bm{\mathrm{w}}}\in[0,1]^n$ and
$\|\bm{x}\|_{1,\bm{\mathrm{w}}}=\sum\limits_{i=1}^n\bm{\mathrm{w}}_{i}|\bm{x}_{i}|$.
The main idea of the weighted $\ell_1$ minimization \eqref{pro2} is
to choose  appropriately the weight vector $\bm{\mathrm{w}}$ such
that in this weighted objective function,  the entries of $\bm{x}$
which are expected to be large are penalized  less. In particular,
the weighted $\ell_1$ minimization \eqref{pro2} reduces to the
standard $\ell_1$ minimization by taking $\bm{\mathrm{w}}=\bm{1}$.
The weighted $\ell_1$ minimization method \eqref{pro2}  has now been
well studied and achieved a complete theoretical system under
various models on the weight vector $\bm{\mathrm{w}}$.
 For example,
in the literature \cite{FMSY,LV, KXAH},  the authors have previously
studied
 the recovery of  signals with prior support information and
 obtained
different conditions to guarantee recovery of these signals via the
weighted $\ell_1$ minimization which only applies a single weight.
Chen and Li \cite{CL}  show that a sharp sufficient recovery
condition based on a high order RIP guarantees stable and robust
recovery of signals via the weighted $\ell_{1}$ minimization
\eqref{pro2} in bounded $\ell_{2}$ and Dantzig selector noise
settings. And the authors  not only point out  that the sufficient
recovery condition is weaker than that of the standard $\ell_{1}$
minimization method  but also point out that the weighted $\ell_{1}$
minimization method gives better upper bounds on the reconstruction
error,  as the accuracy of prior support estimate is at least
$50\%$. Lastly, Needell el.at \cite{NSW} and Chen el.at \cite{CG}
consider the sparse signal recovery with  disjoint prior supports
via the weighted $\ell_1$ minimization method \eqref{pro2} using
arbitrarily many distinct weights
 and obtain the recovery condition and associated recovery guarantees.

 It is well known that
the standard  OMP algorithm as a greedy algorithm is one of the most
effective algorithms in sparse signal recovery because of its
implementation simplicity and competitive recovery performance.
Modifications of the standard OMP  algorithm  have also been studied
for recovering the sparse signals under a partially known support.
As we know, recovering sparse signals with some prior support
information by using OMP algorithm and its modifications is much
fewer than by using weighted $\ell_1$ minimization. Tropp and
Gilbert \cite{TG} first demonstrate theoretically and empirically
sparse signal recovery from
 prior information via a modified OMP algorithm. In \cite{HSIG}, for the noiseless setting the authors  derive a simple recovery guarantee   based on the
mutual coherence of the matrix $\bm{A}$ and the number of true and
wrong indices in prior support $T_0$ for the sparse signal recovery
via the $\mathrm{OMP}_{T_0}$ algorithm  in Table 1. Karahanoglu and
Erdogan \cite{KE1} show that
\begin{align*}
\delta_{k+b+1}<\frac{1}{\sqrt{k-g}+1}
\end{align*}
is sufficient to ensure the sparse signal recovery from
$\bm{y}=\bm{Ax}$ via the $\mathrm{OMP}_{T_0}$, where
$T=\mathrm{supp}(\bm{x})$ with $|T|=k$, $|T\cap T_0|=g$ and
$|T^c\cap T_0|=b$.  However, the above condition on RIP is not
optimal. On the other hand, there is no result considering support
$T\setminus T_0$ recovery via the $\mathrm{OMP}_{T_0}$ algorithm in
the noisy case.

In this paper, we consider optimal sufficient conditions and  some
necessary condition of  the recovery of any
 $k$-sparse signal by the $\mathrm{OMP}_{T_0}$ algorithm in the noiseless and noisy cases.
We consider any $k$-sparse signal $\bm{x}$ with the  prior support
$T_0$, where  the support $T=\mathrm{supp}(\bm{x})$, $|T\cap T_0|=g$
is the number of true indices  and  $|T^c\cap T_0|=b$ is the number
of wrong indices. For the noiseless case, it is shown that
\begin{eqnarray*}
\delta_{k+b+1}<\frac{1}{\sqrt{k-g+1}}
\end{eqnarray*}
ensures  the $\mathrm{OMP}_{T_0}$ algorithm exactly recover the $k$-sparse signal $\bm{x}$ in $k-g$ iterations.
 Moreover, we point out that our condition is sharp in the following sense:
 there exist a sensing  matrix $\bm{A}$ with $\delta_{k+b+1}=\frac{1}{\sqrt{k-g+1}}$, a $k$-sparse signal $\bar{\bm{x}}$  and  the prior support
$T_0$ satisfying $|\mathrm{supp}(\bar{\bm{x}})\cap T_0|=g<k$ and $|(\mathrm{supp}(\bar{\bm{x}}))^c\cap T_0|=b$
such that the $\mathrm{OMP}_{T_0}$ algorithm   may fail to recover the $k$-sparse signal $\bar{\bm{x}}$ in $k-g$ iterations.
For the  noisy case,
we show that if the sensing matrix $\bm{A}$
satisfies $\delta_{k+b+1}<\frac{1}{\sqrt{k-g+1}}$ and $\|v\|_2\leq \varepsilon$, then the $\mathrm{OMP}_{T_0}$ algorithm
exactly recovers the remainder support  $T\setminus T_0$ and obtains the estimated signal $\hat{\bm{x}}$ of the $k$-sparse signal $\bm{x}$  in $k-g$ iterations provided that
\begin{align*}
\min_{i\in T\setminus T_0
}|\bm{x}_i|>\max\Big\{\frac{\sqrt{2(1+\delta_{k+b+1})}\varepsilon}{1-\sqrt{k-g+1}\delta_{k+b+1}},
\ \ \frac{2\varepsilon}{\sqrt{1-\delta_{k+b+1}}}\Big\}.
\end{align*}
Further,
we obtain the upper bounds of $\|\bm{x}-\hat{\bm{x}}\|_2$ and $\max_{i\in T_0\setminus T}|\hat{\bm{x}}_i|$, and
the lower bound of $\min_{i\in T\cap T_0}|\widehat{\bm{x}}_i|$.
At last, we obtain a necessary condition for   exactly recovering the remainder support  $T\setminus T_0$ of the $k$-sparse signal $\bm{x}$
based on the minimum magnitude of elements of
 $\bm{x}_{T\setminus T_0}$. That is,
if the sensing matrix $\bm{A}$ satisfies the RIP of order $k+b+1$ with $0\leq \delta_{k+b+1}<1$ and
the  $\mathrm{OMP}_{T_0}$ algorithm exactly recovers the remainder support  $T\setminus T_0$, then
\begin{align*}
\min_{i\in T\setminus T_0 }|\bm{x}_i| >\frac{\sqrt{1-\delta_{k+b+1}}\varepsilon}{1-\sqrt{k-g+1}\delta_{k+b+1}}.
\end{align*}

The rest of the paper is organized as follows. In Section \ref{2},
we give some notations that will be used throughout this paper, some
significant lemmas and the proofs of them. The main results on the
exact recovery of $k$-sparse signals in the noiseless case and their
proofs are given in Section \ref{3}. Section \ref{4} considers the
exact recovery of the remainder support $T\setminus T_0$ in the
noisy setting. In Section \ref{5}, we discuss the validity of our
sufficient conditions comparing with previous results.

\section{Notations and preliminaries}\label{preliminaries}\label{2}
Let us now define  basic notations. Boldface lowercase letters and boldface uppercase letters respectively denote
column vectors and matrices in the real field $\R$.
 $\langle \cdot,\cdot \rangle$ refers to the inner product
between vectors and $\|\cdot\|_p$ with $p=1,\ 2$ stands for $\ell_p$ norm.
$[n]$ denotes the index set $\{1,2,\ldots,n\}$.
Let $\Gamma\subseteq[n]$ be an index set and $\Gamma^c\subseteq[n]$ be the complementary set of $\Gamma$.  $\bm{x}_\Gamma\in\R^{|\Gamma|}$
denotes the vector composed of components of $\bm{x}\in \R^n$ indexed by $i\in\Gamma$. Define $\tilde{\bm{x}}_\Gamma\in\R^{n}$ by
 \begin{eqnarray*}
(\tilde{\bm{x}}_\Gamma)_i= \left\{
   \begin{array}{ll}
     \bm{x}_i, & \hbox{$i\in \Gamma$;} \\
     0, & \hbox{$i \in \Gamma^c$,}
   \end{array}
 \right.
 \end{eqnarray*}
where $i\in[n].$ Let  the matrix transpose of  the matrix $\bm{A}$
be $\bm{A}^{'}$. $\bm{A}_i$ with $i\in[n]$ denotes the $i$-th column
of $\bm{A}$. Denote by $\bm{A}_\Gamma$ a submatrix of $\bm{A}$
corresponding to $\Gamma$ which consists of  all columns of $\bm{A}$
with index $i\in \Gamma$ . Let $\bm{e}_i\in\R^n$ be the $i$-th
coordinate unit vector.

Let $\bm{A}_{\Gamma}^\dagger$ denote the pseudo-inverse of $\bm{A}_{\Gamma}$. When $\bm{A}_{\Gamma}$ is full column rank ($|\Gamma|\leq m$), then $\bm{A}_{\Gamma}^\dagger=(\bm{A}_{\Gamma}^{'}\bm{A}_{\Gamma})^{-1}\bm{A}_{\Gamma}^{'}$.
 Moreover, $\bm{P}_{\Gamma}=\bm{A}_{\Gamma}\bm{A}_{\Gamma}^\dagger$ and $\bm{P}^\bot_{\Gamma}=\bm{I}-\bm{P}_{\Gamma}$ represent two orthogonal projection operators,
where $\bm{P}_{\Gamma}$  projects a given vector orthogonally onto the spanned space by all columns of $\bm{A}_{\Gamma}$, $\bm{P}^\bot_{\Gamma}$ projects
onto its orthogonal complement and $\bm{I}$ is identity mapping.

The frame of the $\mathrm{OMP}_{T_0}$ algorithm  is formally listed in Table $1$.

\begin{center} Table  1: \ \ \ \ The $\mathrm{OMP}_{T_0} $ algorithm
\end{center}
\hrule
\textbf{Input}\ \ \ \ \ \ \ \ \ measurements $\bm{y}\in\R^m$, sensing matrix $\bm{A}\in\R^{m\times n}$, sparse level $k$,  the number

\ \ \ \ \ \ \ \ \ \ \ \ \ \   of correct indices $g$, prior support $T_0$.

\textbf{Initialize}\ \ \ \ iteration count $t=0$,  estimated  support set $\Lambda_0=T_0$, residual vector

\ \ \ \ \ \ \ \ \ \ \ \ \ \ \ \ \ \   $\bm{r}^{(0)}=\bm{y}-\bm{P}_{\Lambda_0}\bm{y}$.

\hrule
\textbf{While}\ \ \ \ \  stopping criterion is not met \ \ \ \ t=t+1

\ \ \ \ \ \ \ \ \ \ \ \ \ \ \ (Identification\ step)\ \ \ \ $j_t=\arg\max_{i}|\langle \bm{r}^{(t-1)},\bm{A}\bm{e}_i\rangle|$.

\ \ \ \ \ \ \ \ \ \ \ \ \ \ \ (Augmentation step)\ \ \ $\Lambda_t=\Lambda_{t-1}\cup \{j_t\}$.

\ \ \ \ \ \ \ \ \ \ \ \ \ \  \ (Estimation step)\ \ \ \ \ \ \ \ $\bm{x}^{(t)}=\min\limits_{\bm{u}}\|\bm{y}-\bm{A}_{\Lambda_t} \bm{u}\|_2$.

\ \ \ \ \ \ \ \ \ \ \ \ \ \  \ (Residual update step)\ \ \ $\bm{r}^{(t)}=\bm{y}-\bm{A}_{\Lambda_t}\bm{x}^{(t)}$.

\textbf{End}

\textbf{Output}\ \ \ \  the estimated signal
$\hat{\bm{x}}_{\Lambda_t}=\bm{x}^{(t)},\ \ \hat{\bm{x}}_{\Lambda_t^c}=\bm{0}$.
\ \ \\
It is clear that the  $\mathrm{OMP}_{T_0}$ algorithm reduces to the standard OMP algorithm as $T_0=\emptyset$.
In \cite{HSIG}, Herzet et al. give a rigorous definition of ``success'' for the $\mathrm{OMP}_{T_0}$ algorithm, which
matches the classical ``$k$-step'' analysis of the standard OMP algorithm.
\begin{defi}\cite{HSIG}\label{definition2}
The $\mathrm{OMP}_{T_0}$ algorithm with $\bm{y}$ defined in \eqref{model1} as input succeeds if and only if it
selects indices in $T\setminus T_0$ during the first $k-g$ iterations.
\end{defi}

 The authors \cite{HSIG} also proposed the  $\mathrm{OMP}_{T_0}$ algorithm can be understood as
a particular instance of the standard OMP algorithm, in which
indices in the prior support $T_0$ have been identified during the
first $g+b$ iterations. And any condition which guarantees the
success of  the  $\mathrm{OMP}_{T_0}$ algorithm in the sense of
Definition \ref{definition2} ensures the success of the standard OMP
algorithm in $k+b$ iterations provided that the indices in the prior
support
 $T_0$ are selected during the first $g+b$ iterations.

For each iteration of the $\mathrm{OMP}_{T_0}$ algorithm, the solution of  the minimization problem $\min\limits_{\bm{u}}\|\bm{y}-\bm{A}_{\Lambda_t} \bm{u}\|_2$ is
$$\bm{x}^{(t)}=\arg \min\limits_{\bm{u}}\|\bm{y}-\bm{A}_{\Lambda_t} \bm{u}\|_2=\bm{A}_{\Lambda_t}^{\dag}\bm{y}$$
by the least-square method.
Further, by the definition $\bm{A}_{\Lambda_t}^\dagger=(\bm{A}_{\Lambda_t}^{'}\bm{A}_{\Lambda_t})^{-1}\bm{A}_{\Lambda_t}^{'} $ and some simple calculations,
 one has
\begin{align}\label{equation3}
\bm{r}^{(t)}&=\bm{y}-\bm{A}_{\Lambda_t}\bm{x}^{(t)} \nonumber \\
&=\bm{y}-\bm{A}_{\Lambda_t}\bm{A}_{\Lambda_t}^{\dag}\bm{y}\nonumber \\
&=\bm{A}_T\bm{x}_T-\bm{A}_{\Lambda_t}\bm{A}_{\Lambda_t}^{\dag}\bm{A}_T\bm{x}_T+\bm{v}-\bm{A}_{\Lambda_t}\bm{A}_{\Lambda_t}^{\dag}\bm{v}\nonumber \\
&=\bm{A}_{T\setminus\Lambda_t}\bm{x}_{T\setminus \Lambda_t}+\bm{A}_{\Lambda_t}\bm{x}_{\Lambda_t}-\bm{A}_{\Lambda_t}\bm{A}_{\Lambda_t}^{\dag}(\bm{A}_{T\setminus \Lambda_t}\bm{x}_{T\setminus \Lambda_t}+\bm{A}_{\Lambda_t}\bm{x}_{\Lambda_t})+(\bm{I}-\bm{P}_{\Lambda_t})\bm{v}\nonumber \\
&=\bm{A}_{T\setminus\Lambda_t}\bm{x}_{T\setminus\Lambda_t}-\bm{A}_{\Lambda_t}\bm{A}_{\Lambda_t}^{\dag}\bm{A}_{T\setminus\Lambda_t}\bm{x}_{T\setminus\Lambda_t}+(\bm{I}-\bm{P}_{\Lambda_t})\bm{v}\nonumber \\
&=\bm{A}_{T\cup \Lambda_t} \bm{z}_{T\cup \Lambda_t}+(\bm{I}-\bm{P}_{\Lambda_t})\bm{v}
\end{align}
where
\begin{align}\label{equation8}
\bm{z}_{T\cup \Lambda_t}=\left(
                      \begin{array}{c}
                        \bm{x}_{T\setminus\Lambda_t} \\
                        -\bm{A}_{\Lambda_t}^{\dag}\bm{A}_{T\setminus\Lambda_t}\bm{x}_{T\setminus\Lambda_t} \\
                      \end{array}
                    \right).
\end{align}
It is clear that if $T\setminus\Lambda_t \neq \emptyset$ then $\bm{z}_{T\cup \Lambda_t}\neq\bm{0}$.
And $\bm{r}^{(t)}=\bm{P}_{\Lambda_t}^\bot y$, which implies the residual $\bm{r}^{(t)}$ is orthogonal to the columns of $\bm{A}_{\Lambda_t}$.

To analyze the main results of this paper, we establish the
following important lemma.
\begin{lem}\label{lemma1}
Let the support $T=\mathrm{supp}(\bm{x})$ with $|T|=k$ and the prior
support $T_0$ satisfy $|T\cap T_0|=g<k$ and $|T^c\cap T_0|=b$.
Suppose  the sensing matrix $\bm{A}$ satisfies the RIP 　of order
$k+b+1$ and  $\Lambda_t\subseteq T\cup T_0$ for  $0\leq t<k-g$ in
the $\mathrm{OMP}_{T_0}$  algorithm, then
\begin{align}\label{equation2}
&\max_{i\in T\setminus \Lambda_t}|\langle \bm{Ae}_i,\bm{A}_{T\cup \Lambda_t}\bm{z}_{T\cup \Lambda_t}\rangle|
-\max_{i\in (T\cup T_0)^c}|\langle \bm{Ae}_i,\bm{A}_{T\cup \Lambda_t}\bm{z}_{T\cup \Lambda_t}\rangle|\nonumber\\
&\geq\frac{1}{\sqrt{k-g-t}}\bigg(1-\sqrt{k-g-t+1}\delta_{k+b+1}\bigg)\|\tilde{\bm{z}}_{T\cup \Lambda_t}\|_2.
\end{align}
\end{lem}
\begin{proof}
 For simplicity,   let
\begin{eqnarray}\label{e4}
\alpha_1^{(t)}=\max_{i\in T\setminus \Lambda_t}|\langle \bm{Ae}_i,\bm{A}_{T\cup \Lambda_t}\bm{z}_{T\cup \Lambda_t}\rangle|
=\|\bm{A}_{T\setminus\Lambda_t}^{'}\bm{A}_{T\cup \Lambda_t}\bm{z}_{T\cup \Lambda_t}\|_\infty
\end{eqnarray}
and
\begin{align} \label{equation1}
\beta_1^{(t)}
=\max_{i\in (T\cup T_0)^c}|\langle \bm{Ae}_i,\bm{A}_{T\cup \Lambda_t}\bm{z}_{T\cup \Lambda_t}\rangle|
=|\langle \bm{Ae}_{i_t}, \bm{A}_{T\cup \Lambda_t}\bm{z}_{T\cup \Lambda_t}\rangle|
\end{align}
where $i_t=\arg\max\limits_{i\in (T\cup \Lambda_t)^c}|\langle \bm{Ae}_i,\bm{A}_{T\cup \Lambda_t}\bm{z}_{T\cup \Lambda_t}\rangle|$.
Based on the definition of $\alpha_1^{(t)}$ in \eqref{e4}, one obtains that
\begin{eqnarray}\label{e5}
\langle \bm{A}\tilde{\bm{z}}_{T\cup \Lambda_t}, \bm{A}\tilde{\bm{z}}_{T\cup \Lambda_t} \rangle
&=&\langle \bm{A}_{T\cup \Lambda_t}\bm{z}_{T\cup \Lambda_k}, \bm{A}_{T\cup \Lambda_t}\bm{z}_{T\cup \Lambda_t} \rangle\nonumber \\
&=&\langle \bm{z}_{T\cup \Lambda_t}, \bm{A}_{T\cup \Lambda_t}^{'}\bm{A}_{T\cup \Lambda_t}\bm{z}_{T\cup \Lambda_t} \rangle\nonumber\\
&\leq&\|\bm{z}_{T\cup \Lambda_t}\|_2\| \bm{A}_{T\cup \Lambda_t}^{'}\bm{A}_{T\cup \Lambda_t}\bm{z}_{T\cup \Lambda_t}\|_2\nonumber\\
&\stackrel{(1)}{=}&\|\bm{z}_{T\cup \Lambda_t}\|_2\| \bm{A}_{T \setminus \Lambda_t}^{'}\bm{A}_{T\cup \Lambda_t}\bm{z}_{T\cup \Lambda_t}\|_2\nonumber\\
&\leq&\sqrt{k-g-t}\|\bm{z}_{T\cup \Lambda_t}\|_2\| \bm{A}_{T \setminus \Lambda_t}^{'}\bm{A}_{T\cup \Lambda_t}\bm{z}_{T\cup \Lambda_t}\|_\infty\nonumber\\
&=&\sqrt{k-g-t}\|\bm{z}_{T\cup \Lambda_t}\|_2\alpha_1^{(t)}
\end{eqnarray}
where $(1)$ follows from
\begin{align*}
 \bm{A}_{\Lambda_t}^{'}\bm{A}_{T\cup \Lambda_t}\bm{z}_{T\cup \Lambda_t}
& =\bm{A}_{\Lambda_t}^{'}\left(\bm{A}_{T\setminus\Lambda_t}\bm{x}_{T\setminus\Lambda_t}-\bm{A}_{\Lambda_t}\bm{A}_{\Lambda_t}^{\dag}\bm{A}_{T\setminus\Lambda_t}\bm{x}_{T\setminus\Lambda_t}\right)\\
&=\bm{A}_{\Lambda_t}^{'}\bm{A}_{T\setminus\Lambda_t}\bm{x}_{T\setminus\Lambda_t}-\bm{A}_{\Lambda_t}^{'}\bm{A}_{\Lambda_t}(\bm{A}_{\Lambda_t}^{'}\bm{A}_{\Lambda_t})^{-1}
\bm{A}_{\Lambda_t}^{'}\bm{A}_{T\setminus\Lambda_t}\bm{x}_{T\setminus\Lambda_t}\\
&=\bm{0}.
\end{align*}

Let $s=-\frac{\sqrt{k-g-t+1}-1}{\sqrt{k-g-t}}$ and
\begin{eqnarray*}
 \hat{s}_{i_t}=\left\{
       \begin{array}{ll}
         +\|\bm{z}_{T\cup \Lambda_t}\|_2s, & \hbox{$\langle \bm{A}\tilde{\bm{z}}_{T\cup \Lambda_t}, \bm{Ae}_{i_t}\rangle\geq 0$,} \\
         -\|\bm{z}_{T\cup \Lambda_t}\|_2s, & \hbox{$\langle \bm{A}\tilde{\bm{z}}_{T\cup \Lambda_t},\bm{Ae}_{i_t}\rangle<0$,}
       \end{array}
     \right.
\end{eqnarray*}
then
\begin{eqnarray*}
s^2=\frac{\sqrt{k-g-t+1}-1}{\sqrt{k-g-t+1}+1}<1
\end{eqnarray*}
and
 \begin{eqnarray*}
\frac{2\hat{s}_{i_t}}{1-s^2}=\left\{
                               \begin{array}{ll}
                                 -\sqrt{k-g-t}\|\bm{z}_{T\cup \Lambda_t}\|_2, & \hbox{$\langle \bm{A}\tilde{\bm{z}}_{T\cup \Lambda_t}, \bm{Ae}_{i_t}\rangle\geq 0$;} \\
                                 \sqrt{k-g-t}\|\bm{z}_{T\cup \Lambda_t}\|_2, & \hbox{$\langle \bm{A}\tilde{\bm{z}}_{T\cup \Lambda_t}, \bm{Ae}_{i_t}\rangle<0$.}
                               \end{array}
                             \right.
 \end{eqnarray*}
Further, based on  \eqref{e5}, \eqref{equation1} and  some simple calculations we derive that
\begin{align}\label{e6}
&(1-s^4)\sqrt{k-g-t}\|\bm{z}_{T\cup \Lambda_t}\|_2(\alpha_1^{(t)}-\beta_1^{(t)})\nonumber\\
&\geq(1-s^4)\left(\langle \bm{A}\tilde{\bm{z}}_{T\cup \Lambda_t}, \bm{A}\tilde{\bm{z}}_{T\cup \Lambda_t} \rangle
-\sqrt{k-g-t}\|\bm{z}_{T\cup \Lambda_t}\|_2|\langle \bm{Ae}_{i_t},\bm{A}_{T\cup \Lambda_t}\bm{z}_{T\cup \Lambda_t}\rangle|\right)\nonumber\\
&\geq(1-s^4)\left(\langle \bm{A}\tilde{\bm{z}}_{T\cup \Lambda_t}, \bm{A}\tilde{\bm{z}}_{T\cup \Lambda_t} \rangle-\sqrt{k-g-t}\|\tilde{\bm{z}}_{T\cup \Lambda_t}\|_2
|\langle \bm{Ae}_{i_t},\bm{A}\tilde{\bm{z}}_{T\cup \Lambda_t}\rangle|\right)\nonumber\\
&=\|\bm{A}(\tilde{\bm{z}}_{T\cup \Lambda_t}+\hat{s}_{i_t}\bm{e}_{i_t})\|_2^2-\|\bm{A}(s^2\tilde{\bm{z}}_{T\cup \Lambda_t}-\hat{s}_{i_t}\bm{e}_{i_t})\|_2^2.
\end{align}
Because $0\leq t< k-g$, the sensing matrix $\bm{A}$ satisfies  the RIP of order $k+b+1$
with $\delta_{k+b+1}$, $|T\cup \Lambda_t|=k+b$ and $i_t\in (T\cup \Lambda_t)^c$, we obtain that
\begin{eqnarray}\label{e7}
&&\|\bm{A}(\tilde{\bm{z}}_{T\cup \Lambda_t}+\hat{s}_{i_t}\bm{e}_{i_t})\|_2^2-\|\bm{A}(s^2\tilde{\bm{z}}_{T\cup \Lambda_t}-\hat{s}_{i_t}\bm{e}_{i_t})\|_2^2\nonumber\\
&&\geq(1-\delta_{k+b+1})\left(\|\tilde{\bm{z}}_{T\cup \Lambda_t}+\hat{s}_{i_t}\bm{e}_{i_t}\|_2^2\right)
-(1+\delta_{k+b+1})\left(\|s^2\tilde{\bm{z}}_{T\cup \Lambda_t}-\hat{s}_{i_t}\bm{e}_{i_t}\|_2^2\right)\nonumber\\
&&=(1-\delta_{k+b+1})(1+s^2)\|\tilde{\bm{z}}_{T\cup \Lambda_t}\|_2^2-(1+\delta_{k+b+1})(s^4+s^2)\|\tilde{\bm{z}}_{T\cup \Lambda_t}\|_2^2\nonumber\\
&&=\|\tilde{\bm{z}}_{T\cup \Lambda_t}\|_2^2(1+s^2)\left(1-\delta_{k+b+1}-(1+\delta_{k+b+1})s^2\right)\nonumber\\
&&=\|\tilde{\bm{z}}_{T\cup \Lambda_t}\|_2^2(1+s^2)^2\bigg(\frac{1-s^2}{1+s^2}-\delta_{k+b+1}\bigg).
\end{eqnarray}
From the definition of $s$, it follows that
\begin{eqnarray*}\label{e8}
\frac{1-s^2}{1+s^2}=\frac{1-\frac{\sqrt{k-g-t+1}-1}{\sqrt{k-g-t+1}+1}}{1+\frac{\sqrt{k-g-t+1}-1}{\sqrt{k-g-t+1}+1}}=\frac{1}{\sqrt{k-g-t+1}}.
\end{eqnarray*}
Therefore, by \eqref{e6}, \eqref{e7} and the above equality
 we have that
\begin{align*}
\alpha_1^{(t)}-\beta_1^{(t)}&\geq\frac{(1+s^2)^2\bigg(\frac{1-s^2}{1+s^2}-\delta_{k+b+1}\bigg)}{(1-s^4)\sqrt{k-g-t}}\|\tilde{\bm{z}}_{T\cup \Lambda_t}\|_2\\
&=\frac{1}{\sqrt{k-g-t}}\bigg(1-\sqrt{k-g-t+1}\delta_{k+b+1}\bigg)\|\tilde{\bm{z}}_{T\cup \Lambda_t}\|_2.
\end{align*}
\end{proof}

\section{An optimal exact recovery condition in noiseless case}\label{3}
\ \ In this section, we  establish the exact recovery results in
Theorem \ref{the1} and Theorem \ref{the2}. If $j_t\in T\setminus
\Lambda _{t-1}$ ($1\leq t\leq k-g$) in the $t$-th iteration, the
$\mathrm{OMP}_{T_0}$ algorithm  makes a success, i.e.,
$\max\limits_{i\in T\setminus \Lambda_{t-1}}|\langle
\bm{Ae}_i,\bm{r}^{(t-1)}\rangle|
>\max\limits_{i\in (T\cup T_0)^c}|\langle \bm{Ae}_i,\bm{r}^{(t-1)}\rangle|$ in the $t$-th iteration.
Theorem \ref{the1} presents a condition to ensure the exact recovery
of all $k$-sparse signals via the $\mathrm{OMP}_{T_0}$ algorithm in $k-g$
iterations. And we show that our condition is sharp in Theorem
\ref{the2}.
\begin{thm}\label{the1}
Let $\bm{x}\in \R^{n}$ be a $k$-sparse signal in $\bm{y}=\bm{Ax}$, $T$ be the support of $\bm{x}$
with $|T|=k$ and  $T_0$ be a prior  support of $\bm{x}$ satisfying $0\leq|T\cap T_0|=g<k$ and $|T^c\cap T_0|=b$.
Suppose  the sensing matrix $\bm{A}$ satisfies the RIP of order $k+b+1$ with
\begin{eqnarray*}
\delta_{k+b+1}<\frac{1}{\sqrt{k-g+1}}.
\end{eqnarray*}
Then the $\mathrm{OMP}_{T_0}$ algorithm exactly recovers the signal $\bm{x}$ in $k-g$ iterations.
\end{thm}
\begin{proof}
We first prove that under the condition $\delta_{k+b+1}<\frac{1}{\sqrt{k-g+1}}$, the $\mathrm{OMP}_{T_0}$ algorithm
 succeeds in the sense of Definition \ref{definition2} by the inductive method.
For the first iteration, $\Lambda_0=T_0$ and $\bm{r}^{(0)}=\bm{A}_{T\cup T_0}\bm{z}_{T\cup T_0}$. By Lemma \ref{lemma1} with $t=0$ and $\delta_{k+b+1}<\frac{1}{\sqrt{k-g+1}}$,  we have that
\begin{align*}
&\max_{i\in T\setminus T_0}|\langle \bm{Ae}_i,\bm{A}_{T\cup T_0}\bm{z}_{T\cup T_0}\rangle|
-\max_{i\in (T\cup T_0)^c}|\langle \bm{Ae}_i,\bm{A}_{T\cup T_0}\bm{z}_{T\cup T_0}\rangle|\\
&\geq\frac{1}{\sqrt{k-g}}\bigg(1-\sqrt{k-g+1}\delta_{k+b+1}\bigg)\|\tilde{\bm{z}}_{T\cup T_0}\|_2>0
\end{align*}
which means that $\max\limits_{i\in T\setminus T_0}|\langle \bm{Ae}_i, \bm{r}^{(0)}\rangle|
>\max\limits_{i\in (T\cup T_0)^c}|\langle \bm{Ae}_i,\bm{r}^{(0)}\rangle|$.
Then the $\mathrm{OMP}_{T_0}$ algorithm  selects a correct index $j_1\in T\setminus T_0$ in the first iteration.
Suppose that the $\mathrm{OMP}_{T_0}$ algorithm has performed $t$ ($1\leq t <k-g$) iterations successfully,
that is, $\Lambda_{t}\setminus T_0 \subseteq  T\setminus T_0$.
For the $(t+1)$-th iteration, from the equality \eqref{equation3} with $\bm{v}=\bm{0}$, Lemma \ref{lemma1} and $\delta_{k+b+1}<\frac{1}{\sqrt{k-g+1}}$  it follows that
\begin{align*}
&\max_{i\in T\setminus \Lambda_t}|\langle \bm{Ae}_i,\bm{r}^{(t)}\rangle|
-\max_{i\in (T\cup \Lambda_t)^c}|\langle \bm{Ae}_i,\bm{r}^{(t)}\rangle|\\
&=\max_{i\in T\setminus \Lambda_t}|\langle \bm{Ae}_i,\bm{A}_{T\cup \Lambda_t}\bm{z}_{T\cup \Lambda_t}\rangle|
-\max_{i\in (T\cup \Lambda_t)^c}|\langle \bm{Ae}_i,\bm{A}_{T\cup \Lambda_t}\bm{z}_{T\cup \Lambda_t}\rangle|\\
&\geq\frac{1}{\sqrt{k-g-t}}\bigg(1-\sqrt{k-g-t+1}\delta_{k+b+1}\bigg)\\
&\geq\frac{1}{\sqrt{k-g}}\bigg(1-\sqrt{k-g+1}\delta_{k+b+1}\bigg)\\
&>0,
\end{align*}
which implies that the $\mathrm{OMP}_{T_0}$ algorithm make a success in the $(t+1)$-th iteration, i.e.,
 $j_{t+1}\in T\setminus\Lambda_t\subseteq T\setminus T_0$.
Therefore, if $\delta_{k+b+1}<\frac{1}{\sqrt{k-g+1}}$ then the $\mathrm{OMP}_{T_0}$ algorithm succeeds by the Definition \ref{definition2}.

It remains  to prove $\bm{x}=\hat{\bm{x}}$, where $\bm{\hat{x}}$ is the estimated signal of $\bm{x}$  in Table 1.
As the $\mathrm{OMP}_{T_0}$ algorithm has performed $k-g$ iterations successfully, we have that
$\Lambda_{k-g}=T\cup T_0$ and
 \begin{align*}
 \hat{\bm{x}}_{\Lambda_{k-g}}&=\bm{x}^{(k-g)}=\bm{A}^{\dagger}_{\Lambda_{k-g}}\bm{y}\\
 &\stackrel{(1)}{=}(\bm{A}^{'}_{\Lambda_{k-g}}\bm{A}_{\Lambda_{k-g}})^{-1}\bm{A}^{'}_{\Lambda_{k-g}}\bm{A}_T\bm{x}_T\\
 &=(\bm{A}^{'}_{\Lambda_{k-g}}\bm{A}_{\Lambda_{k-g}})^{-1}\bm{A}^{'}_{\Lambda_{k-g}}\bm{A}_{\Lambda_{k-g}}\bm{x}_{\Lambda_{k-g}}
 -(\bm{A}^{'}_{\Lambda_{k-g}}\bm{A}_{\Lambda_{k-g}})^{-1}\bm{A}^{'}_{\Lambda_{k-g}}\bm{A}_{\Lambda_{k-g}\setminus T}\bm{x}_{\Lambda_{k-g}\setminus T}\\
 &\stackrel{(2)}{=}\bm{x}_{\Lambda_{k-g}}
 \end{align*}
 where $(1)$ and $(2)$ respectively follows from the facts that the matrix $\bm{A}$ satisfies the RIP of order $k+b+1$,
which means $\bm{A}_{\Lambda_{k-g}}$ is  full column rank,
 and  $\bm{x}_{\Lambda_{k-g}\setminus T}=\bm{0}$.
 We have completed the proof of the theorem.
\end{proof}

\textbf{Remark}\ \textbf{1.} For any integers $b$ and $g$, the
condition $\delta_{k+b+1}<\frac{1}{\sqrt{k-g+1}}$ is weaker than the
sufficient condition $\delta_{k+b+1}<\frac{1}{\sqrt{k-g}+1}$ in
\cite{KE}.

Next, we show that the condition  $\delta_{k+b+1}<\frac{1}{\sqrt{k-g+1}}$ is optimal in the following theorem.
\begin{thm}\label{the2}
Let $k$ be any given  positive integer, $0\leq g<k$ and $b$ be any
given nonnegative integer. There exist a $k$-sparse signal
$\bar{\bm{x}}$ with $|T|=|\mathrm{supp}(\bar{\bm{x}})|=k$,
 a prior support $T_0$ fulfilling $|T\cap T_0|=g$ and $|T^c\cap T_0|=b$ and a matrix $\bm{A}$ satisfying
\begin{eqnarray*}
\delta_{k+b+1}=\frac{1}{\sqrt{k-g+1}}
\end{eqnarray*}
such that the $\mathrm{OMP}_{T_0}$  algorithm may fail.
\end{thm}
\begin{proof}
For given integers $k>0$, $b\geq0$ and $0 \leq g<k$, let $\bm{A}\in\R^{(k+b+1)\times (k+b+1)}$ be
\begin{eqnarray}\label{matrix1}
\bm{A}=
\begin{pmatrix}
  \ & \  & \  & 0&\cdots&0 &\frac{1}{\sqrt{(k-g+1)(k-g)}}  \\
  \  & \sqrt{\frac{k-g}{k-g+1}}\bm{I}_{k-g} & \  &\vdots&\ &\vdots  & \vdots\\
  \  & \ &\ &0 &\cdots&0& \frac{1}{\sqrt{(k-g+1)(k-g)}}\\
  0 & \cdots& 0& \ & \ &\  & \\
  \vdots  & \ & \vdots& \ & \ &\bm{I}_{g+b+1}&\  \\
   0  &\cdots& 0& \ & \ &\ &\  \\
\end{pmatrix},
\end{eqnarray}
where $\bm{I}_{k-g}$ and $\bm{I}_{g+b+1}$ are  unitary matrices.
Then
 \begin{align*}
 \bm{A}^{'}\bm{A}=
 \begin{pmatrix}
  \ & \  & \  & 0&\cdots&0 &\frac{1}{k-g+1}  \\
  \  & \frac{k-g}{k-g+1}\bm{I}_{k-g} & \  &\vdots&\ &\vdots  & \vdots\\
  \  & \ &\ &0 &\cdots&0& \frac{1}{k-g+1}\\
  0 & \cdots& 0& \ & \ &\  & 0\\
  \vdots  & \ & \vdots& \ & \bm{I}_{g+b} &\ &\vdots \\
   0  &\cdots& 0& \ & \ &\ & 0 \\
    \frac{1}{k-g+1}  &\cdots& \frac{1}{k-g+1} & 0& \cdots&0&1+\frac{1}{k-g+1} \\
\end{pmatrix}.
 \end{align*}
By elementary transformation of determinant, one can verify that
\begin{align*}
&\begin{vmatrix}
  \bm{A}^{'}\bm{A}-\lambda \bm{I}_{k+b+1} \\
\end{vmatrix}\\
&=\begin{vmatrix}
     \ & \  & \  & 0&\cdots&0 &\frac{1}{k-g+1}  \\
  \  & (\frac{k-g}{k-g+1}-\lambda)\bm{I}_{k-g} & \  &\vdots&\ &\vdots  & 0\\
  \  & \ &\ &0 &\cdots&0& \vdots\\
  0 & \cdots& 0& \ & \ &\  & \vdots\\
  \vdots  & \ & \vdots& \ & (1-\lambda)\bm{I}_{g+b} &\ &\vdots \\
   0  &\cdots& 0& \ & \ &\ & 0 \\
    \frac{k-g}{k-g+1}  &\cdots& \frac{1}{k-g+1} & 0& \cdots&0&1+\frac{1}{k-g+1}-\lambda \\
  \end{vmatrix}\\
&=(-1)^{1+1}\bigg(\frac{k-g}{k-g+1}-\lambda \bigg)\bigg(\frac{k-g}{k-g+1}-\lambda\bigg)^{k-g-1}(1-\lambda)^{g+b}\bigg(1+\frac{1}{k-g+1}-\lambda\bigg)\\
&\ \ +(-1)^{k+b+1+1}\frac{k-g}{k-g+1}(-1)^{k+b+1}\frac{1}{k-g+1}\bigg(\frac{k-g}{k-g+1}-\lambda\bigg)^{k-g-1}(1-\lambda)^{g+b}\\
&=(1-\lambda)^{g+b}\bigg(\frac{k-g}{k-g+1}-\lambda \bigg)^{k-g-1}\bigg(\lambda^2-2\lambda+\frac{k-g}{k-g+1}\bigg).
\end{align*}
Then  the eigenvalues $\{\lambda_i\}_{i=1}^{k+b+1}$ of $\bm{A}^{'}\bm{A}$  are
\begin{align*}
&\lambda_1=\cdots=\lambda_{k-g-1}=\frac{k-g}{k-g+1},\ \ \lambda_{k-g}=\cdots=\lambda_{k+b-1}=1,\\
 & \lambda_{k+b}=1-\frac{1}{\sqrt{k-g+1}},\ \ \lambda_{k+b+1}=1+\frac{1}{\sqrt{k-g+1}}.
\end{align*}
Moreover, by definition of the RIP and Remark $1$ in \cite{DM}, the matrix $\bm{A}$ in \eqref{matrix1} satisfies the RIP with
\begin{align*}
\delta_{k+b+1}&=\max\{1-\lambda_{\min}(\bm{A}^{'}\bm{A}),\ \ \lambda_{\max}(\bm{A}^{'}\bm{A})-1\}\\
&=\max\{1-\lambda_{k+b},\ \ \lambda_{k+b+1}-1\}=\frac{1}{\sqrt{k-g+1}}.
\end{align*}

Consider $k$-sparse signal $\bar{\bm{x}}=(\underbrace{1,\cdots,1,}_{k}0,\cdots,0)'\in\R^{k+b+1}$
 and the prior support $T_0=\{k-g+1,\cdots,k,k+1,\cdots,k+b\}$.
For the first iteration,
\begin{align*}
 \bm{r}^{(0)}&=\bm{A}_{T\setminus T_0}\bar{\bm{x}}_{T\setminus T_0}-\bm{A}_{T_0}(\bm{A}_{T_0}^{'}
\bm{A}_{T_0})^{-1}\bm{A}_{T_0}^{'}\bm{A}_{T\setminus T_0}\bar{\bm{x}}_{T\setminus T_0}\\
 &=(\underbrace{\sqrt{\frac{k-g}{k-g+1}},\cdots,\sqrt{\frac{k-g}{k-g+1}}}_{k-g},0,\cdots,0)^{'}\in\R^{k+b+1}.
 \end{align*}
 In fact,
 \begin{align*}
\bm{A}_{T\setminus T_0}\bar{\bm{x}}_{T\setminus T_0}=(\underbrace{\sqrt{\frac{k-g}{k-g+1}},\cdots,\sqrt{\frac{k-g}{k-g+1}}}_{k-g},0,\cdots,0)^{'}\in\R^{k+b+1},
 \end{align*}
and $\bm{A}_{T_0}^{'}\bm{A}_{T\setminus T_0}\bar{\bm{x}}_{T\setminus T_0}=\bm{0}\in\R^{g+b}$.

 For $i\in T\setminus T_0$, we have
 \begin{eqnarray*}
 |\langle \bm{Ae}_i, \bm{r}^{(0)}\rangle|=\frac{k-g}{k-g+1}.
 \end{eqnarray*}
For $i\in (T\cup T_0)^c=\{k+b+1\}$, it follows immediately that
\begin{align*}
 |\langle \bm{Ae}_i, \bm{r}^{(0)}\rangle|=\frac{k-g}{k-g+1}.
\end{align*}
It is obvious that $\max\limits_{i\in T\setminus T_0}\langle
\bm{Ae}_i,\bm{r}_0\rangle=\max\limits_{i\in(T\cup T_0)^c}\langle
\bm{Ae}_i,\bm{r}_0\rangle$ which implies the $\mathrm{OMP}_{T_0}$
algorithm  may fail to identify one index of the subset $T\setminus
T_0$ in the first iteration. So the $\mathrm{OMP}_{T_0}$ algorithm
may fail for the given matrix $\bm{A}$, the $k$-sparse signal
$\bar{\bm{x}}$ and the prior support $T_0$.
\end{proof}
\section{ Analysis on the remainder  support $T\setminus T_0$ recovery in noisy case }\label{4}
\ \
In this section, we respectively establish sufficient conditions and a necessary condition  for the exact remainder support $T\setminus T_0$ recovery
of the $k$-sparse signal $\bm{x}$ with the prior
 support $T_0$ in the model \eqref{model1} with $\bm{v}\neq\bm{0}$
 via the $\mathrm{OMP}_{T_0}$ algorithm within $k-g$ iterations.  In such case,  since the exact reconstruction of the $k$-sparse signal $\bm{x}$ cannot
be guaranteed, we use the upper bound of $\|\bm{x}-\hat{\bm{x}}\|_2$
as a performance measure of the $\mathrm{OMP}_{T_0}$ algorithm  and
obtain the upper bound. In order to recover the  whole support $T$,
we investigate  the upper bound of $\max\limits_{i\in T_0\setminus
T}|\bm{x}_i|$ and
 the lower bound of $\min\limits_{i\in T\cap T_0}|\bm{x}_i|$.
Here, we only consider $l_2$ bounded noise, i.e., $\|\bm{v}\|_2\leq\varepsilon$.

\subsection{Sufficient conditions  for the remainder support $T\setminus T_0$ recovery}

In Theorem \ref{the5},
our conditions are in terms of the RIP of order $k+b+1$ and the minimum magnitude of the  entries of   $\bm{x}_{T\setminus T_0}$.
The upper bounds of $\max\limits_{i\in T_0\setminus T}|\bm{x}_i|$ and $\|\bm{x}-\hat{\bm{x}}\|_2$ and the lower bound of $\min\limits_{i\in T\cap T_0}|\bm{x}_i|$
are obtained in Theorem \ref{the6}.

\begin{thm}\label{the5}
Let $\bm{x}$ be a $k$-sparse signal in the model \eqref{model1}, $T$ be the support of the signal $\bm{x}$ with $|T|=k$ and  $T_0$ be a prior  support of the signal $\bm{x}$
such that $|T\cap T_0|=g<k$ and $|T^c\cap T_0|=b$.
Suppose $\|\bm{v}\|_2\leq\varepsilon$  and  the sensing matrix $\bm{A}$ satisfies
\begin{eqnarray}\label{equation13}
\delta_{k+b+1}<\frac{1}{\sqrt{k-g+1}}.
\end{eqnarray}
Then  the $\mathrm{OMP}_{T_0}$ algorithm  with the stopping rule $\|\bm{r}^{(t)}\|_2\leq\varepsilon$
 exactly recovers the remainder support $T\setminus T_0$ of the signal $\bm{x}$ in  $k-g$ iterations provided that
\begin{eqnarray}\label{equation14}
\min_{i\in T\setminus T_0 }|\bm{x}_i|>\max\Big\{\frac{\sqrt{2(1+\delta_{k+b+1})}\varepsilon}{1-\sqrt{k-g+1}\delta_{k+b+1}},\ \ \frac{2\varepsilon}{\sqrt{1-\delta_{k+b+1}}}\Big\}.
\end{eqnarray}
\end{thm}

\begin{proof} The proof consists of two parts. In the first part
we show that the $\mathrm{OMP}_{T_0}$ algorithm selects  indices of the remainder support $T\setminus T_0$ in each  iteration
 under conditions \eqref{equation13} and \eqref{equation14}.
In the second part we prove that the $\mathrm{OMP}_{T_0}$ algorithm exactly performs $|T\setminus T_0|=k-g$ iterations
with the stopping rule $\|\bm{r}^{(t)}\|_2\leq\varepsilon$.

Part $I$:
By mathematical induction method,
suppose  first that the $\mathrm{OMP}_{T_0}$ algorithm  performed $t$ ($1 \leq t < k-g$) iterations successfully, that is,
$\Lambda_t\subseteq T\cup T_0$ and $j_1,\cdots, j_t \in T\setminus T_0$.
Then by the $\mathrm{OMP}_{T_0}$ algorithm in Table 1,  we need to show $j_{t+1} \in T\setminus \Lambda_t$ which
means the $\mathrm{OMP}_{T_0}$ algorithm makes a success in the $(t+1)$-th iteration. By the fact that
$\bm{r}^{(t)}$ is orthogonal to each column of $\bm{A}_{\Lambda_t}$,
we only need to prove that
\begin{align} \label{equation4}
\max_{i\in T\setminus\Lambda_t}|\langle \bm{Ae}_i, \bm{r}^{(t)}\rangle|>\max_{i\in (T\cup T_0)^c}|\langle \bm{Ae}_i, \bm{r}^{(t)}\rangle|
\end{align}
for the $(t+1)$-th iteration.

From  \eqref{equation3}, one has that
\begin{align}\label{equation5}
\max_{i\in T\setminus\Lambda_t}|\langle \bm{Ae}_i, \bm{r}^{(t)}\rangle|
&=\max_{i\in T\setminus\Lambda_t}|\langle \bm{Ae}_i, \bm{A}\tilde{\bm{z}}_{T\cup\Lambda_t}+\bm{P}_{\Lambda_t}^\perp \bm{v}\rangle|\nonumber\\
&\geq\max_{i\in T\setminus\Lambda_t}|\langle \bm{Ae}_i, \bm{A}\tilde{\bm{z}}_{T\cup\Lambda_t}\rangle|-\max_{i\in T\setminus\Lambda_t}|\langle \bm{Ae}_i, \bm{P}_{\Lambda_t}^\perp \bm{v}\rangle|
\end{align}
and
\begin{align}\label{equation6}
\max_{i\in (T\cup T_0)^c}|\langle \bm{Ae}_i, \bm{r}^{(t)}\rangle|
&=\max_{i\in (T\cup T_0)^c}|\langle \bm{Ae}_i, \bm{A}\tilde{\bm{z}}_{T\cup\Lambda_t}+\bm{P}_{\Lambda_t}^\perp \bm{v}\rangle|\nonumber\\
&\leq\max_{i\in (T\cup T_0)^c}|\langle \bm{Ae}_i, \bm{A}\tilde{\bm{z}}_{T\cup\Lambda_t}\rangle|+\max_{i\in (T\cup T_0)^c}|\langle \bm{Ae}_i, \bm{P}_{\Lambda_t}^\perp \bm{v}\rangle|.
\end{align}
Therefore, by \eqref{equation5} and \eqref{equation6}, it suffices to prove
that
\begin{align}\label{equation7}
\max_{i\in T\setminus\Lambda_t}|\langle \bm{Ae}_i, \bm{A}\tilde{\bm{z}}_{T\cup\Lambda_t}\rangle|-\max_{i\in (T\cup T_0)^c}|\langle \bm{Ae}_i, \bm{A}\tilde{\bm{z}}_{T\cup\Lambda_t}\rangle|\nonumber\\
>\max_{i\in T\setminus\Lambda_t}|\langle \bm{Ae}_i, \bm{P}_{\Lambda_t}^\perp \bm{v}\rangle|+\max_{i\in (T\cup T_0)^c}|\langle \bm{Ae}_i, \bm{P}_{\Lambda_t}^\perp \bm{v}\rangle|.
\end{align}

One  first gives a lower bound on the left-hand side of \eqref{equation7}. From Lemma \ref{lemma1},
the definition of $\bm{z}_{T\cup T_0}$ in \eqref{equation8} and the induction
 assumption $j_1,\cdots,j_t\in T\setminus T_0$ which implies $|T\setminus\Lambda_t|=k-g-t$, it follows that
\begin{align}\label{equation10}
&\max_{i\in T\setminus \Lambda_t}|\langle \bm{Ae}_i,\bm{A}_{T\cup \Lambda_t}\bm{z}_{T\cup \Lambda_t}\rangle|
-\max_{i\in (T\cup T_0)^c}|\langle \bm{Ae}_i,\bm{A}_{T\cup \Lambda_t}\bm{z}_{T\cup \Lambda_t}\rangle|\nonumber\\
&\geq\frac{1}{\sqrt{k-g-t}}\bigg(1-\sqrt{k-g-t+1}\delta_{k+b+1}\bigg)\|\tilde{\bm{z}}_{T\cup \Lambda_t}\|_2\nonumber\\
&\geq\frac{1}{\sqrt{k-g-t}}\bigg(1-\sqrt{k-g-t+1}\delta_{k+b+1}\bigg)\|\bm{x}_{T\setminus\Lambda_t}\|_2\nonumber\\
&\geq\frac{1}{\sqrt{k-g-t}}\bigg(1-\sqrt{k-g-t+1}\delta_{k+b+1}\bigg)\sqrt{k-g-t}\min_{i\in T\setminus\Lambda_t}|\bm{x}_i|\nonumber\\
&\geq\bigg(1-\sqrt{k-g+1}\delta_{k+b+1}\bigg)\min_{i\in T\setminus T_0}|\bm{x}_i|.
\end{align}
One  now gives an upper bound on the right-hand side of \eqref{equation7}. There exist the indices $i^{(t)}\in T\setminus \Lambda_t$ and $i_1^{(t)}\in (T\cup T_0)^c$
 satisfying
 \begin{align*}
\max_{i\in T\setminus\Lambda_t}|\langle \bm{Ae}_i, \bm{P}_{\Lambda_t}^\perp \bm{v}\rangle|=|\langle \bm{Ae}_{i^{(t)}}, \bm{P}_{\Lambda_t}^\perp \bm{v}\rangle|
\end{align*}
and
\begin{align*}
\max_{i\in (T\cup T_0)^c}|\langle \bm{Ae}_i, \bm{P}_{\Lambda_t}^\perp \bm{v}\rangle|=|\langle \bm{Ae}_{i_1^{(t)}}, \bm{P}_{\Lambda_t}^\perp \bm{v}\rangle|,
\end{align*}
respectively.
Therefore, we obtain that
\begin{align}\label{equation9}
&\max_{i\in T\setminus\Lambda_t}|\langle \bm{Ae}_i, \bm{P}_{\Lambda_t}^\perp \bm{v}\rangle|+\max_{i\in (T\cup T_0)^c}|\langle \bm{Ae}_i, \bm{P}_{\Lambda_t}^\perp \bm{v}\rangle|\nonumber\\
&=|\langle \bm{Ae}_{i^{(t)}}, \bm{P}_{\Lambda_t}^\perp \bm{v}\rangle|+|\langle \bm{Ae}_{i_1^{(t)}}, \bm{P}_{\Lambda_t}^\perp \bm{v}\rangle|\nonumber\\
&=\| \bm{A}_{\{i^{(t)},i_1^{(t)}\}}^{'} \bm{P}_{\Lambda_t}^\perp \bm{v}\|_1\nonumber\\
&\leq\sqrt{2}\| \bm{A}_{\{i^{(t)},i_1^{(t)}\}}^{'} \bm{P}_{\Lambda_t}^\perp \bm{v}\|_2\nonumber\\
&\stackrel{(1)}\leq\sqrt{2(1+\delta_{k-g+1})}\|\bm{P}_{\Lambda_t}^\perp \bm{v}\|_2\nonumber\\
&\stackrel{(2)}\leq\sqrt{2(1+\delta_{k-g+1})}\varepsilon
\end{align}
where (1) follows from $\bm{A}$ fulfilling the RIP with order
$k-g+1$ ($g<k$) and $(2)$ is because the fact
\begin{align*}
\|\bm{P}_{\Lambda_t}^\perp \bm{v}\|_2\leq\|\bm{P}_{\Lambda_t}^\perp \|_2\|\bm{v}\|_2\leq\|\bm{v}\|_2\leq\varepsilon.
\end{align*}
By \eqref{equation13} and \eqref{equation14}, there is
\begin{align*}
\bigg(1-\sqrt{k-g+1}\delta_{k+b+1}\bigg)\min_{i\in T\setminus T_0}|\bm{x}_i|>\sqrt{2(1+\delta_{k-g+1})}\varepsilon.
\end{align*}
It is obvious that \eqref{equation7} holds by the above inequality. Then the $\mathrm{OMP}_{T_0}$ algorithm  selects one  index from the subset $T\setminus \Lambda_t$
in the $(t+1)$-th iteration.
In conclusion, we have shown that the $\mathrm{OMP}_{T_0}$ algorithm selects one   index from $T\setminus T_0$ in each iteration.

 Part $II$: We prove that  the $\mathrm{OMP}_{T_0}$ algorithm  performs exactly $k-g$ iterations. It remains to show
that $\|\bm{r}^{(t)}\|_2>\varepsilon$ for $0\leq t <k-g$ and $\|\bm{r}^{(k-g)}\|_2\leq\varepsilon$.

Since the $\mathrm{OMP}_{T_0}$ algorithm  selects an index of $T\setminus T_0$
  in each iteration under the  conditions
\eqref{equation13} and \eqref{equation14}, $\Lambda_{k-g}=T\cup T_0$ which means $\bm{P}^{\perp}_{\Lambda_{k-g}}\bm{A}_T\bm{x}_T=\bm{0}$.
Moreover,
\begin{align*}
\|\bm{r}^{(k-g)}\|_2=\|\bm{P}^{\perp}_{\Lambda_{k-g}}\bm{A}_T\bm{x}_T+\bm{P}^{\perp}_{\Lambda_{k-g}}\bm{v}\|_2
=\|\bm{P}^{\perp}_{\Lambda_{k-g}}\bm{v}\|_2\leq\|\bm{v}\|_2\leq\varepsilon.
\end{align*}
For $0\leq t<k-g$, we have that $\Lambda_t\subseteq T\cup T_0$,
$(T\cup T_0)\setminus \Lambda_{t}\neq\emptyset$ and
\begin{align*}
\|\bm{r}^{(t)}\|_2&=\|\bm{A}_{T\cup \Lambda_{t}}\bm{z}_{T\cup\Lambda_t}+(\bm{I}-\bm{P}_{\Lambda_t})\bm{v}\|_2\\
&\geq\|\bm{A}_{T\cup \Lambda_{t}}\bm{z}_{T\cup\Lambda_t}\|_2-\|\bm{P}^{\perp}_{\Lambda_{t}}\bm{v}\|_2\\
&\stackrel{(1)}\geq \sqrt{1-\delta_{k+b}}\|\bm{z}_{T\cup\Lambda_t}\|_2-\varepsilon\\
&\geq \sqrt{1-\delta_{k+b+1}}\|\bm{x}_{T\setminus\Lambda_t}\|_2-\varepsilon\\
&\geq \sqrt{1-\delta_{k+b+1}}\sqrt{k-g-t}\min_{T\setminus\Lambda_t}|\bm{x}_i|-\varepsilon\\
&\geq \sqrt{1-\delta_{k+b+1}}\min_{T\setminus T_0}|\bm{x}_i|-\varepsilon\\
&\stackrel{(2)}>\varepsilon
\end{align*}
where (1) is because $\bm{A}$ satisfies the  RIP with order $k+b+1$ and $\|\bm{P}_{\Lambda_t}^\perp \bm{e}\|_2\leq\varepsilon$ and
 (2) is because of \eqref{equation14}.
We have completed the proof.
\end{proof}
\begin{thm}\label{the6}
Let $\bm{x}$ be a $k$-sparse signal in the model \eqref{model1} with $\|\bm{v}\|_2\leq \varepsilon$.
 $T$ be the support of $\bm{x}$ with $|T|=k$ and  $T_0$ be a prior  support of $\bm{x}$
such that $|T\cap T_0|=g<k$ and $|T^c\cap T_0|=b$. If $\delta_{k+b+1}<\frac{1}{\sqrt{k-g+1}}$,
\begin{eqnarray}\label{equation27}
\min_{i\in T}|\bm{x}_i|>\max\Big\{\frac{\sqrt{2(1+\delta_{k+b+1})}\varepsilon}{1-\sqrt{k-g+1}\delta_{k+b+1}},\ \ \frac{2\varepsilon}{\sqrt{1-\delta_{k+b+1}}}\Big\}.
\end{eqnarray}
and  the stopping rule $\|\bm{r}^{(t)}\|_2\leq\varepsilon$, then
\begin{align*}
\min_{i\in T\cap T_0}|\hat{\bm{x}}_i| > \frac{\varepsilon}{\sqrt{1-\delta_{k+b+1}}},\ \ \ \
\max_{i\in T_0\setminus T}|\hat{\bm{x}}_i|\leq\frac{\varepsilon}{\sqrt{1-\delta_{k+b+1}}}
\end{align*}
and
\begin{align*}
\|\bm{x}-\bm{\hat{x}}\|\leq\frac{\varepsilon}{\sqrt{1-\delta_{k+b+1}}},
\end{align*}
where $\hat{\bm{x}}$ is the estimated signal of $\bm{x}$ in Table 1.
\end{thm}
\begin{proof}
It is obvious that the condition \eqref{equation14} is satisfied by \eqref{equation27}.
From  Theorem \ref{the5},  the  condition  $\delta_{k+b+1}<\frac{1}{\sqrt{k-g+1}}$ and the the lower bound \eqref{equation27} ensure
the $\mathrm{OMP}_{T_0}$ algorithm with the stopping rule $\|\bm{r}^{(t)}\|_2\leq\varepsilon$  exactly stops after  performing  $k-g$ iterations successfully, which implies $\Lambda_{k-g}=T\cup T_0$.
For the $\mathrm{OMP}_{T_0}$ algorithm in Table $1$, there exists
\begin{align*}
\bm{x}^{(k-g)}=\arg\min_{\bm{u}}\|\bm{y}-\bm{A}_{\Lambda_{k-g}}\bm{u}\|=\bm{A}^{\dag}_{T\cup T_0}\bm{y}=\bm{A}^{\dag}_{T\cup T_0}(\bm{A}_{T\cup T_0}\bm{x}_{T\cup T_0}+\bm{v})
= \bm{x}_{T\cup T_0}+\bm{\omega}
\end{align*}
where
\begin{align*}
\bm{\omega}=(\bm{A}_{T\cup T_0}^{'}\bm{A}_{T\cup T_0})^{-1}\bm{A}_{T\cup T_0}^{'}\bm{v}.
\end{align*}
Furthermore, we have that
\begin{align*}
\hat{\bm{x}}_i=\left\{
                 \begin{array}{ll}
                   \bm{x}_i+\bm{\omega}_i, & \hbox{$i\in T$,} \\
                   \bm{\omega}_i, & \hbox{$i\in T_0\setminus T$,}\\
                   0,& \hbox{$i\in (T_0\cup T)^c$,}
                 \end{array}
               \right.
\end{align*}
and
\begin{align*}
\sqrt{1-\delta_{k+b+1}}\|\bm{\omega}\|_2\leq \|\bm{A}_{T\cup T_0}\bm{\omega}\|_2=\|\bm{P}_{T\cup T_0}\bm{v}\|_2\leq \|\bm{v}\|\leq \varepsilon.
\end{align*}
Therefore,  by \eqref{equation27} and the above equalities and inequality, we obtain that
\begin{align*}
\min_{i\in T\cap T_0}|\hat{\bm{x}}_i| \geq \min_{i\in
T}(|\bm{x}_i|-|\bm{\omega}_i|)
> \frac{\varepsilon}{\sqrt{1-\delta_{k+b+1}}},
\end{align*}
\begin{align*}
\max_{i\in T_0\setminus T}|\hat{\bm{x}}_i|=\max_{i\in T_0\setminus T}|\bm{\omega}_i|
\leq \|\bm{\omega}\|_2\leq\frac{\varepsilon}{\sqrt{1-\delta_{k+b+1}}}
\end{align*}
and
\begin{align*}
\|\bm{x}-\hat{\bm{x}}\|_2&\leq \frac{1}{\sqrt{1-\delta_{T\cup T_0}}}\|\bm{A}(\bm{x}-\hat{\bm{x}})\|_2\\
&=\frac{1}{\sqrt{1-\delta_{T\cup T_0}}}\|\bm{A}_{T\cup T_0}\bm{x}_{T\cup T_0}-\bm{A}_{T\cup T_0} \bm{x}^{(k-g)}\|_2\\
&=\frac{1}{\sqrt{1-\delta_{T\cup T_0}}}\|\bm{A}_{T\cup T_0}\bm{\omega}\|_2\\
&\leq \frac{\varepsilon}{\sqrt{1-\delta_{k+b+1}}}.
\end{align*}
\end{proof}
\subsection{A necessary condition  for the remainder support $T\setminus T_0$ recovery}

In this subsection, we derive a necessary condition on the minimum
magnitude of
 the  components of $\bm{x}_{T\setminus T_0}$ for the exact recovery of the remainder support $T\setminus T_0$.

\begin{thm}\label{them1}
Let $\bm{x}$ be a $k$-sparse signal in the model \eqref{model1},
$T$ be the support of $\bm{x}$ with $|T|=k$ and  $T_0$ be a prior  support of the $\bm{x}$ such that $|T\cap T_0|=g<k$ and $|T^c\cap T_0|=b$.
Suppose $\|\bm{v}\|_2\leq\varepsilon$  and  the sensing matrix $\bm{A}$ satisfies
the RIP of order $k+b+1$ with $0\leq\delta_{k+b+1}<1$.
If  the $\mathrm{OMP}_{T_0}$ algorithm
 exactly recovers the remainder support $T\setminus T_0$ of the signal $\bm{x}$ in  $k-g$ iterations,
then
\begin{align}\label{equation15}
\min_{T\setminus T_0}|\bm{x}_i|>\frac{\sqrt{1-\delta_{k+b+1}}\varepsilon}{1-\sqrt{k-g+1}\delta_{k+b+1}}.
\end{align}
\end{thm}
\begin{proof}
The proof  below roots in \cite{WZWTM}. However, some essential
modifications are necessary in order to adapt the results to sparse
signals $\bm{x}$ with the prior  support $T_0$. Using proofs by
contradiction, we show the theorem. We construct a linear model of
the form $\bm{y}=\bm{Ax}+\bm{v}$,
 where the sensing matrix  $\bm{A}$ and the error vector $\bm{v}$
respectively  satisfy the RIP of order $k+b+1$ with
$0\leq\delta_{k+b+1}(\bm{A})=\delta_{k+b+1}<1$ and
$\|\bm{v}\|_2\leq\varepsilon$, and $\bm{x}$ is a $k$-sparse signal
with the prior support $T_0$ and satisfies
\begin{align}\label{equation16}
\min_{T\setminus T_0}|\bm{x}_i|\leq\theta :=\frac{\sqrt{1-\delta_{k+b+1}}\varepsilon}{1-\sqrt{k-g+1}\delta_{k+b+1}},
\end{align}
such that the  $\mathrm{OMP}_{T_0}$ algorithm may fail to
 exactly recover the remainder support $T\setminus T_0$ of the  signal $\bm{x}$  within  $k-g$ iterations.

 It is well known that there exist the unit vectors  $\bm{\xi}^{(1)},\bm{\xi}^{(2)},\cdots,\bm{\xi}^{(k-g-1)}\in\R^{k-g}$ such that the matrix
 \begin{align*}
 \left(
   \begin{array}{c}
     \bm{\xi}^{(1)}\ \ \bm{\xi}^{(2)}\ \ \cdots\ \ \bm{\xi}^{(k-g-1)}\ \ \frac{1}{\sqrt{k-g}}\bm{1}_{k-g} \\
   \end{array}
 \right)\in \R^{{(k-g)}\times{(k-g)}}
\end{align*}
is orthogonal, which implies $\langle \bm{\xi}^{(i)}, \bm{\xi}^{(j)}\rangle=0$ and $\langle\bm{\xi}^{(i)}, \bm{1}_{k-g}\rangle=0$
for $i,j=1,\cdots,k-g-1$ and $i\neq j$, where $\bm{1}_{k-g}=(1,\cdots,1)^{'}\in\R^{k-g}$.
Let the matrix
\begin{align}\label{equation17}
\bm{U}^{'}=\left(
  \begin{array}{cccccccc}
    \bm{\xi}^{(1)} & \cdots & \bm{\xi}^{(k-g-1)}& \frac{\bm{1}_{k-g}}{\sqrt{(k-g)(\eta^2+1)}}  & \ &\bm{0}_{(k-g)\times(g+b)}&\ &\frac{\eta \bm{1}_{k-g}}{\sqrt{(k-g)(\eta^2+1)}} \\
    0 & \cdots &0 & 0 & \ &\ &\ & 0 \\
    \vdots & \  & \vdots  & \vdots & \ &\bm{I}_{g+b}&\ & \vdots \\
    0 & \cdots & 0 & 0 & \ &\ &\ & 0  \\
    0 & \cdots & 0 & \frac{\eta}{\sqrt{\eta^2+1}} & 0 & \cdots& 0 &-\frac{1}{\sqrt{\eta^2+1}}\\
  \end{array}
\right),
\end{align}
 where
 \begin{align*}
 \eta=\frac{\sqrt{k-g+1}-1}{\sqrt{k-g}}.
 \end{align*}
 Then
 $\bm{U}$ is also an orthogonal matrix.

 Let $\bm{D}\in\R^{(k+b+1)\times(k+b+1)}$ be a diagonal matrix with
 \begin{align}\label{equation18}
 d_{ii}=\left\{
          \begin{array}{ll}
            \sqrt{1-\delta_{k+b+1}}, & \hbox{$i=k-g$,} \\
            \sqrt{1+\delta_{k+b+1}}, & \hbox{$i\neq k-g$,}
          \end{array}
        \right.
 \end{align}
and the sensing matrix $\bm{A}=\bm{DU}$,
then $\bm{A}^{'}\bm{A}=\bm{U}^{'}\bm{D}^2\bm{U}$.
In the following, we show that $\delta_{k+b+1}(\bm{A})=\delta_{k+b+1}$.
For any $\bm{x}\in \R^{k+b+1}$, setting $\hat{\bm{\nu}}=\bm{Ux}$, we have that
\begin{align*}
\|\bm{Ax}\|_2^2&=\langle \bm{Ax}, \bm{Ax} \rangle
=\bm{x}^{'}\bm{A}^{'}\bm{Ax}=(\bm{Ux})^{'}\bm{D}^{'}\bm{D}(\bm{Ux})\\
&=\hat{\bm{\nu}}^{'}\bm{D}^2\hat{\bm{\nu}}
=(1+\delta_{k+b+1})\|\hat{\bm{\nu}}\|_2^2-2\delta_{k+b+1}\hat{\bm{\nu}}_{k-g}^2\\
&\leq (1+\delta_{k+b+1})\|\hat{\bm{\nu}}\|_2^2
\stackrel{(1)}=(1+\delta_{k+b+1})\|\bm{x}\|_2^2
\end{align*}
and
\begin{align*}
\|\bm{Ax}\|_2^2&=\langle \bm{Ax}, \bm{Ax} \rangle
=(1-\delta_{k+b+1})\|\hat{\bm{\nu}}\|_2^2+2\delta_{k+b+1}\sum_{1\leq i\leq k+b+1,\ i\neq k-g}\hat{\bm{\nu}}_i^2\\
&\geq (1-\delta_{k+b+1})\|\hat{\bm{\nu}}\|_2^2
\stackrel{(2)}=(1-\delta_{k+b+1})\|\bm{x}\|_2^2,
\end{align*}
where (1) and (2)  result of the fact that $\bm{U}$ is an orthogonal matrix.
Then, based on the definition \ref{definition1},
we have $\delta_{k+b+1}(\bm{A})\leq \delta_{k+b+1} $.
It  remains to prove that the matrix $\bm{A}=\bm{DU}$ satisfies $\delta_{k+b+1}(\bm{A})\geq\delta_{k+b+1}.$
Let the vector
\begin{align*}
\hat{\bm{x}}=((\bm{\xi}^{(1)})^{'},0,\cdots,0)^{'}\in\R^{k+b+1},
\end{align*}
then $\hat{\bm{x}}$ is $(k+b+1)$-sparse and $\|\hat{\bm{x}}\|_2^2=1$. By the definitions of $\bm{D}$ and $\bm{A}$,
we obtain that
\begin{align*}
\|\bm{A}\hat{\bm{x}}\|_2^2=(\bm{U}\hat{\bm{x}})^{'}\bm{D}^2\bm{U}\hat{\bm{x}}=\bm{e}_1^{'}\bm{D}^2\bm{e}_1=1+\delta_{k+b+1}=(1+\delta_{k+b+1})\|\hat{\bm{x}}\|_2^2.
\end{align*}
So $\delta_{k+b+1}(\bm{A})\geq\delta_{k+b+1}$.
In conclusion, $\delta_{k+b+1}(\bm{A})=\delta_{k+b+1}$.

Let the original signal
\begin{align*}
\bm{x}=(\overbrace{\theta,\cdots, \theta,}^{T\setminus T_0}\overbrace{1,\cdots,1,}^{T_0\cap T}\overbrace{0,\cdots,0}^{T_0\setminus T},0)^{'}\in\R^{k+b+1},
\end{align*}
where $\theta$ is defined in \eqref{equation16}. Then the signal $\bm{x}$
is $k$-sparse with the support $T=\{1,2,\cdots,k\}$, the prior support $T_0=\{k-g+1,\cdots,k+b\}$ and satisfies \eqref{equation16}.
It is not hard to prove that $\bm{A}_{T\setminus T_0}=\bm{DU}_{T\setminus T_0}$.
Moreover, by some simple calculations we derive that
\begin{align*}
&\bm{A}_{T\setminus T_0}\bm{x}_{T\setminus T_0}=\bm{DU}_{T\setminus T_0}\bm{x}_{T\setminus T_0}\\
&=\bm{D}\bigg(\underbrace{0,\cdots,0}_{k-g-1},\sqrt{\frac{k-g}{\eta^2+1}}\theta,\underbrace{0,\cdots, 0}_{g+b},
\sqrt{\frac{k-g}{\eta^2+1}}\eta \theta\bigg)^{'}\\
&=\sqrt{1+\delta_{k+b+1}}\bigg(\underbrace{0,\cdots, 0}_{k-g-1},\sqrt{\frac{(k-g)(1-\delta_{k+b+1})}{(\eta^2+1)(1+\delta_{k+b+1})}}\theta ,
\underbrace{0, \cdots,0}_{g+b}, \sqrt{\frac{k-g}{\eta^2+1}}\eta \theta\bigg)^{'}.
\end{align*}
and
\begin{align}\label{equation24}
\bm{A}^{'}\bm{A}_{T\setminus T_0}\bm{x}_{T\setminus T_0}=(\underbrace{\mu,\cdots,\mu,}_{k-g}\underbrace{0,\cdots,0,}_{g+b} -\frac{2\eta}{\eta^2+1}\sqrt{k-g}\delta_{k+b+1}\theta)^{'}
\end{align}
where $\mu=\frac{(1-\delta_{k+b+1})+(1+\delta_{k+b+1})\eta^2}{\eta^2+1}\theta$.
Similarly, let the error vector
\begin{align*}
\bm{v}&=\bm{D}^{-1}\bm{U}(\underbrace{0,\cdots,0}_{k+b},-\sqrt{1-\delta_{k+b+1}}\varepsilon)^{'},\\
&=\bm{D}^{-1}(\underbrace{0, \cdots,0}_{k-g-1},\frac{-\sqrt{1-\delta_{k+b+1}}\eta\varepsilon}{\sqrt{\eta^2+1}},\underbrace{0,\cdots, 0}_{g+b},\sqrt{\frac{1-\delta_{k+b+1}}{\eta^2+1}}\varepsilon)^{'}\\
&=(\underbrace{0,\cdots,0}_{k-g-1},\frac{-\eta\varepsilon}{\sqrt{\eta^2+1}}, \underbrace{0,\cdots ,0}_{g+b},\sqrt{\frac{1-\delta_{k+b+1}}{(\eta^2+1)(1+\delta_{k+b+1})}}\varepsilon)^{'}
\end{align*}
then $\|\bm{v}\|_2\leq\varepsilon$,
\begin{align}\label{equation25}
\bm{A}^{'}\bm{v}&=\bm{U}^{'}\bm{DD}^{-1}\bm{U}(\underbrace{0,\cdots,0}_{k+b}, -\sqrt{1-\delta_{k+b+1}}\varepsilon)^{'}\nonumber\\
&=(\underbrace{0,\cdots,0}_{k+b}, -\sqrt{1-\delta_{k+b+1}}\varepsilon)^{'}.
\end{align}
 By \eqref{equation24} and \eqref{equation25}, it is clear that
\begin{align*}
\bm{A}_{T_0}^{'}\bm{A}_{T\setminus T_0}\bm{x}_{T\setminus T_0}=\bm{0},\ \ \ \
\bm{A}_{T_0}^{'}\bm{v}=\bm{0}.
\end{align*}
Therefore, using \eqref{equation3} and the above equality, we obtain that
\begin{align*}
\bm{r}^{(0)}&=\bm{A}_{T\setminus T_0}\bm{x}_{T\setminus T_0}-\bm{A}_{T_0}(\bm{A}_{T_0}^{'}{\bm{A}_{T_0}})^{-1}\bm{A}_{T_0}^{'}\bm{A}_{T\setminus T_0}\bm{x}_{T\setminus T_0}+\bm{v}-\bm{A}_{T_0}(\bm{A}_{T_0}^{'}{\bm{A}_{T_0}})^{-1}\bm{A}_{T_0}^{'}\bm{v}\\
&=\bm{A}_{T\setminus T_0}\bm{x}_{T\setminus T_0}+\bm{v}.
\end{align*}
Therefore, we have that
\begin{align*}
\langle \bm{Ae}_i, \bm{r}^{(0)} \rangle&=\left\{
                            \begin{array}{ll}
                              \frac{(1-\delta_{k+b+1})+(1+\delta_{k+b+1})\eta^2}{\eta^2+1}\theta, & \hbox{$i\in T\setminus T_0$} \\
                              -\frac{2\eta}{\eta^2+1}\sqrt{k-g}\delta_{k+b+1}\theta-\sqrt{1-\delta_{k+b+1}}\varepsilon, & \hbox{$i=k+b+1$}
                            \end{array}
                          \right.\\
&=\left\{
    \begin{array}{ll}
      (1-\frac{1}{\sqrt{k-g+1}}\delta_{k+b+1})\theta, & \hbox{$i\in T\setminus T_0$} \\
      -\frac{k-g}{\sqrt{k-g+1}}\delta_{k+b+1}\theta-\sqrt{1-\delta_{k+b+1}}\varepsilon, & \hbox{$i=k+b+1$.}
    \end{array}
  \right.
\end{align*}
From \eqref{equation16}, it follows that
\begin{align*}
\max_{i\in T\setminus T_0}|\langle \bm{Ae}_i, \bm{r}^{(0)}
\rangle|=\max_{i\in(T\cup T_0)^c}|\langle \bm{A e}_i, \bm{r}^{(0)}
\rangle|,
\end{align*}
which means  the $\mathrm{OMP}_{T_0}$  algorithm may choose a wrong index $k+b+1$ in the first iteration.
That is, the remainder support $T\setminus T_0$ of the signal $\bm{x}$ may not be exactly recovered in $k-g$ iterations
 by the $\mathrm{OMP}_{T_0}$  algorithm.
We completed the proof.
\end{proof}

\section{Discussion}\label{5}
In this section, we shall focus exclusively the discussions on the
validity of our sufficient condition. In section \ref{3}, for any
$k$-sparse signals $\bm{x}$ with $|T|=|\mathrm{supp}(\bm{x})|=k$
from $\bm{y}=\bm{Ax}$ and the prior support $T_0$ satisfying $|T\cap
T_0|=g<k$ and $|T_0\setminus T|=b$, we have established the
condition based on the RIC $\delta_{k+b+1}<\frac{1}{\sqrt{k-g+1}}$
to guarantee the exact recovery  of the signal $\bm{x}$ via the
$\mathrm{OMP}_{T_0}$ algorithm in $k-g$ iterations and proved the
upper bound of RIC depending  on $g$ is sharp. It is known from
Theorem III.1 in  \cite{M} that if $\bm{A}$ satisfies the condition
$\delta_{k+1}<\frac{1}{\sqrt{k+1}}$ then the standard OMP algorithm
will recover any $k$-sparse signals $\bm{x}$ from $\bm{y}=\bm{Ax}$
in $k$ iterations. Moreover, the author \cite{M} also show that the
condition $\delta_{k+1}<\frac{1}{\sqrt{k+1}}$  is sharp.
 In order to state the validity of the sharp condition in this paper, we need to compare the
 two bounds
 \begin{align}\label{5-1}
 \delta_{k+b+1}<\frac{1}{\sqrt{k-g+1}}
 \end{align}
 and
 \begin{align}\label{5-2}
 \delta_{k+1}<\frac{1}{\sqrt{k+1}}.
 \end{align}
Since $\delta_{k+b+1}\geq \delta_{k+1}$ and
$\frac{1}{\sqrt{k-g+1}}\geq\frac{1}{\sqrt{k+1}}$, it is impossible
to compare these two sharp conditions directly. Intuitively, when
$b$ is very small and $g$ is large, we expect that the sharp
condition \eqref{5-1} to be weaker than  the condition \eqref{5-2}.
For example, taking  $b=0$ and $0<g<k$, the condition \eqref{5-1} is
weaker than the  condition \eqref{5-2}. Now, we establish exact
comparison of these two bounds of $\delta_{k+b+1}$ in \eqref{5-1}
 and  $\delta_{k+1}$ in \eqref{5-2} for some particular cases in the following theorem.
 \begin{thm}
  For any positive  integers $c\geq 3$,
 assume that $k>2c^2-1$, $(1-\frac{1}{c^2})(k+1)\leq g <k$ and $1\leq b\leq (c-2)\lceil\frac{k}{2}\rceil$,
  then the condition  $\delta_{k+b+1}<\frac{1}{\sqrt{k-g+1}}$ in this paper is weaker than the sufficient condition $\delta_{k+1}<\frac{1}{\sqrt{k+1}}$ \cite{M}.
 \end{thm}
 \begin{proof}
  By $g\geq (1-\frac{1}{c^2})(k+1)$, we derive that
 \begin{align}\label{4-1}
 \frac{c}{\sqrt{k+1}}\leq\frac{1}{\sqrt{k-g+1}}.
 \end{align}
Since $1 \leq b\leq (c-2)\lceil\frac{k}{2}\rceil$, we have $k+b+1\leq c\lceil\frac{k}{2}\rceil$.
Then, $\delta_{k+b+1}\leq \delta_{c\lceil\frac{k}{2}\rceil}$.
Therefore, from $\delta_{cr}<c\cdot \delta_{2r}$ for any positive integers $c$ and $r$ (seeing Corollary 3.4 in \cite{NT}),
 the fact $k+1\geq 2\lceil\frac{k}{2}\rceil$ with $k\geq 2$, $\delta_{k+1}<\frac{1}{\sqrt{k+1}}$ and the inequality \eqref{4-1},
 it follows that
 \begin{align*}
 \delta_{k+b+1}\leq \delta_{c\lceil\frac{k}{2}\rceil}<c\delta_{2\lceil\frac{k}{2}\rceil}\leq c\delta_{k+1}< \frac{c}{\sqrt{k+1}}\leq \frac{1}{\sqrt{k-g+1}},
 \end{align*}
 which implies the condition $\delta_{k+b+1}$ in this paper is weaker than the sufficient condition $\delta_{k+1}<\frac{1}{\sqrt{k+1}}$.
 We complete the proof of the theorem.
\end{proof}


\end{document}